\documentclass[letterpaper, 10 pt, conference]{ieeeconf}

\IEEEoverridecommandlockouts                             

\title{\LARGE \bf Signalling and Control  in  Nonlinear Stochastic  Systems: An Information State Approach  with   Applications }

\usepackage{amsfonts}                                                          
\usepackage{epsfig}

\usepackage{amsthm}
\usepackage{graphicx}        
\usepackage{latexsym}
\usepackage{amssymb}
\usepackage{amsmath}
\usepackage{parskip}
\usepackage{multirow}
\usepackage{cite}
\usepackage{mathtools}
\usepackage{epstopdf}
\usepackage{etoolbox}
\usepackage{array}
\usepackage{mathptmx}
\usepackage[bottom]{footmisc}
\usepackage{tikz}

\usetikzlibrary{decorations.pathreplacing}

\usepackage{verbatim}
\usepackage{calrsfs}
\usepackage{booktabs}

\DeclareMathAlphabet{\pazocal}{OMS}{zplm}{m}{n}

\newcommand{\mb}{\mathbb}

\let\bbordermatrix\bordermatrix
\patchcmd{\bbordermatrix}{8.75}{4.75}{}{}
\patchcmd{\bbordermatrix}{\left(}{\left[}{}{}
\patchcmd{\bbordermatrix}{\right)}{\right]}{}{}
\renewcommand{\theequation}{\thesection.\arabic{equation}}
\renewcommand{\thefigure}{\thesection.\arabic{figure}}
\renewcommand{\thetable}{\thesection.\arabic{table}}

\newcommand{\sr}{\stackrel}

\newcommand{\rar}{\rightarrow}

\newcommand{\tri}{\sr{\triangle}{=}}

\newcommand{\be}{\begin{equation}}
\newcommand{\ee}{\end{equation}}
\newcommand{\bea}{\begin{eqnarray}}
\newcommand{\eea}{\end{eqnarray}}
\newcommand{\bes}{\begin{eqnarray*}}
\newcommand{\ees}{\end{eqnarray*}}
\newcommand{\bce}{\begin{center}}
\newcommand{\ece}{\end{center}}
\newcommand{\beae}{\begin{IEEEeqnarray}{rCl}}
\newcommand{\eeae}{\end{IEEEeqnarray}}

\def\VR{\kern-\arraycolsep\strut\vrule &\kern-\arraycolsep}
\def\vr{\kern-\arraycolsep & \kern-\arraycolsep}

\newcommand{\ben}{\begin{enumerate}}
\newcommand{\een}{\end{enumerate}}

\newcommand{\hso}{\hspace{.1in}}
\newcommand{\hst}{\hspace{.2in}}

\newtheorem{theorem}{Theorem}[section]

\newtheorem{remark}{Remark}[section]

\newtheorem{assumptions}{Assumptions}[section]
\newtheorem{definition}{Definition}[section]
\newtheorem{lemma}{Lemma}[section]

\author{Charalambos D. Charalambous$^{1}$ and Stelios Louka$^{2}$
\thanks{$^{1}$C. D. Charalambous (Bambos) is  with the Faculty of  Department of Electrical and Computer Engineering, University of Cyprus, Nicosia, Cyprus
        {\tt\small chadcha@ucy.ac.cy}}%
\thanks{$^{2}$S. Louka is a Ph.D. student  with the Department of Electrical and Computer Engineering, University of Cyprus, Nicosia, Cyprus
        {\tt\small slouka01@ucy.ac.cy}}%
        }

\begin{document}

\maketitle
\thispagestyle{empty}
\pagestyle{empty}

\begin{abstract}
We consider optimal signalling and control of  discrete-time nonlinear partially observable stochastic systems in state space form.
In the first part of the paper,  we characterize the operational  {\it control-coding capacity}, $C_{FB}$ in bits/second,  by an  information theoretic optimization problem  of encoding signals or messages into randomized  controller-encoder strategies,  and reproducing the messages at the output of the system using a decoder or estimator with arbitrary small  asymptotic error probability.  Our analysis of  $C_{FB}$ is based on realizations of   randomized strategies (controller-encoders), in terms of information states of   nonlinear filtering theory, and either uniform or arbitrary distributed  random variables (RVs). In the second part of the paper,   we analyze  the linear-quadratic  Gaussian partially observable stochastic system (LQG-POSS). We show that  simultaneous signalling and control  leads  to  randomized strategies  described by finite-dimensional sufficient statistics, that involve two Kalman-filters,   and  consist of  control, estimation  and signalling strategies. We apply decentralized optimization techniques to prove a separation principle, and  to derive  the optimal  control part of randomized strategies explicitly in terms of a control matrix  difference Riccati equation (DRE). 
\end{abstract}

%

\section{Introduction, Main Results, Literature}
\label{sect:problem}
In this   paper we analyze the  problem of simultaneous signalling  and control for   partially observable discrete-time nonlinear stochastic control systems,  with respect to an average pay-off. 
 Our protocol for simultaneous signalling and control is a slight generalization of Shannon's \cite{shannon1948} operation definition 
for  reliable communication of  signals or messages over  noisy channels, subject to a total transmitter power constraint of $\kappa \in [0,\infty)$ units of power, called  {\it coding rate and  coding capacity}, in the theory of information transmission  \cite{gallager1968,ihara1993,cover-thomas2006}. 
By  analogy to Shannon  \cite{shannon1948}, we quantify  an achievable {\it control-coding (CC) rate}  $R$ in bits/second,  by  the number of messages that can be encoded into  randomized control strategies (consisting of  controller,  encoder and estimation strategies) subject to the average pay-off constraint,   and reproduced by the   decoder or estimator at the output of the stochastic control system,   with arbitrary small asymptotic error probability. We call the  maximum achievable CC  rate  the    {\it control-coding  capacity}, $C_{FB}(\kappa)$  bits/second, of the stochastic control system. Embedded  in the operational definition of an  achievable CC rate is the dual role of randomized control strategies, \\
(1) to control the state of the stochastic control system while  meeting  the  average pay-off constraint at each time instant and asymptotically, and \\
(2) to signal the messages encoded  into the randomized control strategies using the stochastic control system as a communication channel,   and to reproduce the messages at the output of the decoder,  with arbitrary small asymptotic error probability. \\
 Our first  goal   is to characterize $C_{FB}(\kappa)$ as function of $\kappa \in [0,\infty)$,  by an information theoretic optimization problem over randomized strategies subject to the average pay-off constraint,  using nonlinear filtering theory and sufficient statistics. The analysis of 
$C_{FB}(\kappa)$   includes  as degenerate cases, many problems of   classical stochastic optimal  control theory.  In particular, for signalling to occur it is necessary that  $C_{FB}(\kappa)\in (0,\infty)$ for $\kappa \in (\kappa_{min}, \infty)$, where  $\kappa_{min}$  is precisely the minimum  asymptotic average pay-off  over   deterministic control  strategies (i.e., nonrandomized) of  the  partially observable  nonlinear stochastic system  \cite{kumar-varayia1986} \cite{caines1988,bensoussan:B1992,elliott-aggoun-moore1995}. \\
Our  second goal is to  characterize $C_{FB}(\kappa)$ for   partially observable stochastic 
Linear-Quadratic-Gaussian (LQG) control systems.
The analysis of  of $C_{FB}(\kappa)$ includes as degenerate cases,  many noisy communication  channels with memory,  often used   in communication applications, such as, channels with Gaussian noise represented by state space models \cite{wolfowitz1975,cover-pombra1989,yang-kavcic-tatikonda2005,yang-kavcic-tatikonda2007,kim2010,gattami2019,charalambous-kourtellaris-tziortzis:SICON-2024}, molecular communication channels  \cite{galluccio:IEEEC2018,farsad:IEEECST016}, communication channels with finite states known to the encoder and/or the decoder \cite{gallager1968,chen-berger2005,tatikonda-mitter2009,permuter-cuff-roy-weissman2010,elishco-permuter2014,permuter-asnani-weissman2013,chen-berger2005}, etc. \\
In the rest of this section, 
we present the  problem formulation,   a brief summary of our main results, followed by their  relations to past literature.

\subsection{Problem Formulations and Main Results}
\label{sect:pr-mr}
\subsubsection{Nonlinear Partially Observable Stochastic System (N-POSS)}
\label{sub-section:A.1}
We consider a N-POSS  with  
 inputs $A^n\tri \{A_1, \ldots, A_n\}$,    $A_t:\Omega \rar  {\mathbb A}$, outputs $Y^n\tri \{Y_1, \ldots, Y_n\}$, $Y_t : \Omega \rar {\mathbb Y}$,  and states $X^n\tri \{X_1, \ldots, X_n\}$, $X_t: \Omega \rar {\mathbb X}$,  with  conditional distributions,    $\forall t \in  {\mathbb Z}_+^n \tri \{1,2, \ldots, n\}$: 
\begin{align}
&{\mb P}\big\{X_{t+1} \in dx_{t+1} \big| Y_{t},X_t, A_t\big\}={\bf P}_{X_{t+1}|Y_t, X_t, A_t}   \nonumber \\
& \hst =S_{t+1}(dx_{t+1}|Y_t, X_t, A_t),  \label{N-DM_1} \\
&{\mb P}\big\{Y_{t} \in dy_{t} \big| X_{t}, A_{t}\big\}={\bf P}_{Y_{t}|X_{t}, A_{t}}=Q_t(dy_t|X_t, A_t). 
 \label{N-DM_2}
\end{align}
These  distributions are  induced by    the nonlinear  recursive equations subject to an average cost constraint: 
\begin{align}
&X_{t+1}=f_t(X_t, A_t, W_t), \hso X_1=x_1, \hso \forall t \in {\mathbb Z}_+^{n-1}, \label{NDM-3}\\
& Y_{t}   = h_t(X_t, A_t, W_t),\hso \forall t \in {\mathbb Z}_+^{n},  \label{NDM-4} \\
& \frac{1}{n}{\bf E}\big\{c_{n}(A^n,X^n) \big\} \leq \kappa \in [0, \infty),\label{NDM-5} 
\end{align}
\begin{align}
{\bf P}_{X_1, W^n}= P_{X_1} \times \prod_{t=1}^n{\bf P}_{W_t} \hso \mbox{($W^n$ indep. and inde. of $X_1$)}.
 \label{NDM-6}
\end{align}
Here,  $W^n\tri \{W_1, \ldots, W_n\}$, $W_t: \Omega \rar {\mathbb W}$ is  the  noise process, and $f_t(\cdot), h_t(\cdot),  c_n(\cdot), \forall t$ are measurable functions. The cost function $c_n(\cdot)$ depends on  $(A^n, X^n)$ to reflect the cost of controlling $(X^n,Y^n)$ via $A^n$.
Throughout the paper,  unless otherwise stated,  $({\mathbb A}, {\mathbb Y}, {\mathbb X},  {\mathbb W})$, are  abstract Borel spaces,  which  include    finite-alphabet spaces,  finite-dimensional Euclidean spaces, etc. 
The N-POSS  may   correspond to stable/unstable control or communication channels with state $X^n$
\cite{caines1988}.    

{\it Operational Signalling and Control Protocol.} Given the N-POSS, we consider a  code  denoted by  $  \{(n, {\cal M}^{(n)}, \epsilon_n, \kappa)|n=1, \dots\}$,
as  follows.  \\
(a) Uniformly distr. messages $M : \Omega \rar  {\cal M}^{(n)} \tri  \{1, 2,\ldots, M^{(n)} \}$.\\
(b) Controller-Encoder  strategies $g^n(\cdot)\tri (g_1(\cdot),\ldots, g_n(\cdot)) $ mapping messages $M$, and feedback information   into inputs\footnote{ $(A_t=A_t^g, Y_t=Y^g, X_t=X_t^g)$ since these are generated from  $g(\cdot)$.},  $A_1=g_1(M), A_2=g_2(M, g_1(M), Y_1), A_3=g_3(M, g_1(M), g_2(M, g_1(M), Y_1), Y^2), \ldots$,    of block length $n$, 
 defined by\footnote{The superscript on expectation operator ${\bf E}^g$  indicates that the corresponding distribution ${\bf P}= {\bf P}^g$  depends the encoding strategy $g$.} 
\begin{align}
&{\cal E}_{n}(\kappa) \triangleq  \Big\{g_t: {\cal M}^{(n)}  \times   {\mathbb A}^{t-1} \times {\mb Y}^{t-1}  \rar {\mb A}_t, \:A_1=g_1(M), \nonumber \\
&A_2=g_2(M,g_1(M),Y_1),\ldots, A_n=g_n(M, g^{n-1}(M, Y^{n-2}), Y^{n-1})\Big|\nonumber \\
& \mbox{condition   {\bf (C1)}) holds,} \hso   \frac{1}{n}   {\bf E}^g    \big\{c_n  (A^n, X^n) \big\} \leq \kappa  \Big\},  \label{NCM-6}\\
&\mbox{\bf (C1)} \hso   {\bf P}_{A_t|A^{t-1}, Y^{t-1}, W^t, X^t}= {\bf P}_{A_t|A^{t-1}, Y^{t-1}}, \;   \forall t \in {\mathbb Z}_+^n.   \label{ci_ci_1}
\end{align}
Conditional independence {\bf (C1) } is equivalent to the Markov chain (MC), $(W^t, X^t) \leftrightarrow (A^{t-1},Y^{t-1}) \leftrightarrow A_t, \forall t \in {\mathbb Z}_+^n$. We  assume $(X_1, W^n)$ is independent of the messages $M$. 
The information to the  controller-encoder strategies  is  ${\cal I}_t^g\tri \{M,A^{t-1}, Y^{t-1}\},  A_t=g_t({\cal I}_t^g), \forall t \in {\mathbb Z}_+^n$. \\
(c)  Decoder function,  $y^n \longmapsto d_{n}(y^n)\in  {\cal M}^{(n)}$,  with average  error  probability ${\bf P}_{e}^{(n)} = \frac{1}{M^{(n)}} \sum_{i=1}^{M^{(n)}} {\mathbb  P}\big\{d_n(Y^n) \neq i|M=i\big\}= \epsilon_n\in [0,1]$.\\
A rate $R$ is called  an {\it achievable CC rate}, if there exists  a controller-encoder and a decoder  sequence satisfying
$\lim_{n\longrightarrow\infty} {\epsilon}_n=0$ and $\liminf_{n \longrightarrow\infty}\frac{1}{n}\log{M}^{(n)}\geq R$. The operational definition of the {\it CC  capacity}  is $C_{FB}(\kappa) \triangleq \sup \{R\big| R \: \: \mbox{is achievable}\}$. 

{\it Main Results of  N-POSS.} Our main results include   the following.\\
1.1) Analysis of   the sequential information theoretic characterization  of $C_{FB}$ under  Dobrushin's information stability  \cite{pinsker1964}  defined by (see  \cite{charalambous-kourtellaris-tziortzis:SICON-2024})
\begin{align}
&C_{FB}(\kappa) \tri \lim_{n \longrightarrow \infty}\frac{1}{n}C_{n,FB}(\kappa),  \\
&C_{FB,n}(\kappa)=  \sup_{{\cal P}_n\tri  \{P_t(da_t|a^{t-1}, y^{t-1})\}_ {t=1}^n, \;  \frac{1}{n}{\bf E}^P\{ c_n(A^n,X^n) \} \leq \kappa}I(A^n\rar Y^n),
\end{align}
\begin{align}
&I^P(A^n\rar Y^n) \tri  \sum_{t=1}^n I^P(A^t; Y_t|Y^{t-1}), \\
&I^P(A^t; Y_t|Y^{t-1})= H^P(Y_t|Y^{t-1})-H^P(Y_t|Y^{t-1}, A^t)
\end{align}
where $H^P(\cdot|\cdot)$ denotes conditional (differential)  entropy \cite{ihara1993}.

{\it Cost-Rate Optimization Problem.} A  dual optimization    problem to $ C_{FB,n}(\kappa)$ is the cost-rate optimization  problem defined by 
 \begin{align}
\kappa_{n}(C) \triangleq   \inf_{{\cal P}_n, \: \mbox{such that} \:  \frac{1}{n} I(A^n \rar Y^n) \: \geq \: C} 
 {\bf  E}^P\big\{c_{n}(A^n, X^{n})\big\} \label{cap_fb_1_TC_n_n}
 \end{align}
 where $C \in [0, \infty]$.
By  \cite{charalambous-kourtellaris-tziortzis:SICON-2024},     $\kappa_{n}(C)$ is a convex   non-decreasing  function in $C \in (0, \infty)$. Next,  we relate  $\kappa_{n}(C)$ to classical stochastic optimal control. 

{\it Relation to Classical Stochastic Optimal Control Problems.}
 Let ${\cal E}_{n}^{D}$ denote the    restriction of randomized strategies ${\cal P}_{n}$ to the set of deterministic strategies ${\cal E}_{n}^{D} \triangleq  \big\{ a_j= e_j(a^{{j-1}},y^{{j-1}}) | j=1,\ldots, n \big\}$.
Define the two stochastic optimal control problems with randomized and deterministic strategies, respectively,    by 
\begin{align*}
J_{n}^{SC}(P^*)\tri  \inf_{{\cal P}_{n}} 
 {\bf  E}^{P } \big\{c_{n}(A^n, X^{n})\big\}, \; J_{n}^{SC}(e^*)\tri  \inf_{ {\cal E}_{n}^{D}} 
 {\bf  E}^e \big\{c_{n}(A^n, X^n)\big\}. 
\end{align*}
Using the fact that deterministic  strategies in $J_{n}^{SC}(e^*)$ (when they exists) achieve the performance of  randomized strategies in $J_{n}^{SC}(P^*)$ (see  \cite{gihman-skorohod1979}),  we also have  (by $I(A^n \rar Y^n)\geq 0$),
\begin{align*}
\frac{1}{n}\kappa_{n}(C)  \geq \frac{1}{n}\kappa_{n}(C)\big|_{C=0}=\frac{1}{n} J_{n}^{SC}(P^*)=\frac{1}{n}J_{n}^{SC}(e^*) =\kappa_{n,min}. 
\end{align*}
Consequently, $\kappa_{n,min} \in [0,\infty)$ is the  minimum cost required to control the N-POSS with zero signalling rate.  The cost of information signalling at a rate $C\in (0,\infty)$   is $
\kappa(C)- \kappa(0)\tri \lim_{n \longrightarrow \infty} \frac{1}{n} \kappa_{n}(C)- \lim_{n \longrightarrow \infty} \frac{1}{n} \kappa_{n}(0)
$
provided the supremums and  limits exist and they are finite. For a non-zero CC  rate,  $C_{FB}(\kappa) \tri \lim_{n \longrightarrow \infty}\frac{1}{n}C_{FB,n}(\kappa)\in (0,\infty)$,    it is necessary that   $\kappa$  exceeds  the critical value   $\kappa_{min}\tri \kappa(0)$. Hence, $\kappa_{min}=\lim_{n \longrightarrow \infty}\frac{1}{n} J_{n}^{SC}(e^*)$  is precisely the asymptotic minimum cost of the classical stochastic optimal control with partial observations (i.e., corresponding to zero   information signalling).

Our analysis is  based on generalizations of the concepts of  the   "information state" and "sufficient statistic" often used in stochastic optimal control theory with deterministic or nonrandomized strategies   \cite{striebel1965,pmortensen66,kumar-varayia1986,charalambous1997role}.  
 In particular, for the analysis of $C_{FB, n}(\kappa)$ (and hence $C_{FB}(\kappa)$) we  show the following main results. \\
{\it 1.2) Theorem~\ref{theorem:aposteriori}.} A sufficient statistic for the input distribution is (i) the   information state defined by the a posteriori conditional distribution $\xi_t(A^{t-1}, Y^{t-1})\tri {\bf P}_{X_t|A^{t-1}, Y^{t-1}}, \forall t\in {\mathbb Z}_+^n$,  and  (ii)  independent uniformly distributed  RVs $U^n\tri \{U_1, \ldots, U_n\},   U_t: \Omega \rar {\mathbb U}\tri [0,1]$,  with distribution ${\bf P}_{U_i}$,   in the sense that, there exists  functions $\mu_t^u(\cdot)$ such that 
\begin{align}
&A_t={\mu}_t^u(\xi_t(A^{t-1}, Y^{t-1}), U_t),\hso  \forall t\in {\mathbb Z}_+^n, \label{stra_un}\\
&{ P}_t(da_t| a^{t-1},y^{t-1}) = {\bf P}_{U_t} \big( u_t \in {\mathbb U}\big|{\mu}_t^u(\xi_t(a^{t-1}, y^{t-1}), u_t)  \in da_t \big).\nonumber 
\end{align}
We also show that strategies  (\ref{stra_un}) can be replaced by another function  $\mu^z(\cdot, Z_t)$, with $U^n$ replaced by arbitrary independent RVs $Z^n, Z_t :\Omega \rar {\mathbb Z}$    with distribution ${\bf P}_{Z_i}$.\\
{\it 1.3) Theorem~\ref{thm:IL}.} Equivalent expressions of  $C_{FB,n}(\kappa)$ using $\mu^u(\cdot)$  and $\mu^z(\cdot)$, with  $I^P(A^n\rar y^n)=I^{\mu^u}(U^n \rar  Y^n)= I^{\mu^z}(Z^n \rar  Y^n)$.  

\vspace{0.1cm}
\subsubsection{Linear-Quadratic-Gaussian Decision Model (LQG-POSS)}
 \label{sub-section:A.2}
 We analyze  (the multiple-input multiple output (MIMO))   LQG-POSS  with $A_t :  \Omega \rar {\mathbb R}^{n_a}$,  $X_t :  \Omega \rar {\mathbb R}^{n_x}$,   $Y_t :  \Omega \rar {\mathbb R}^{n_y}$, correlated noise $V_t :  \Omega \rar {\mathbb R}^{n_y}$,   $(n_a,n_y,n_v, n_x)$ positive integers,  defined by
\begin{align}
&Y_t=D_tA_t + V_t,  \ \  V_t= C_t X_{t} +  N_t W_t, \ \   \forall t \in {\mathbb Z}_+^{n},  \label{LQG-1}\\
&X_{t+1}=F_{t} X_{t}+ B_t A_t+  G_{t} W_t, \hso \forall t \in {\mathbb Z}_+^{n-1}, \label{LQG-3}\\
 & X_1\in G(\mu_{X_1},K_{X_1}), \hso K_{X_1} \succeq 0 \hso \mbox{(pos. semi-definite)}, \label{LQG-4} \\
&W_t\in G(0,K_{W_t}),\hso K_{W_t} \succeq 0, \hso N_t K_{W_t} N_t^T \succ 0, \hso \forall t \in {\mathbb Z}_+^{n} \label{LQG-5}\\
& \frac{1}{n}{\bf E}\big\{c_{n}(A^n,X^n)\big\} \leq \kappa,  \hso  c_{n}(a^n,x^n) \tri \sum_{t=1}^n \gamma_i(a_t, x_t), \label{LQG-6} \\
&\gamma_t(a_t,x_{t}) \tri \langle a_t, R_{t} a_t \rangle+ \langle x_t, Q_{t} x_t \rangle, 
 \;
  Q_t\succeq 0, \;R_t \succ 0, \; \forall t
 \label{LQG-7}
\end{align}
where 
$X\in G(\mu_X, K_X)$ means the RV $X$ is  Gaussian with distribution ${\bf P}_{X}$,  having  mean $\mu_X$ and covariance $K_X$, and $\left\langle \cdot, \cdot \right\rangle$ denotes inner product.
We show the following main results.\\
{\it 2.1) Lemma~\ref{lemma_POSS}, Lemma~\ref{lemma_Gaussian}.}  $C_{FB,n}(\kappa)$ is a maximum entropy  problem, because  ${\bf P}_{X_t|A^{t-1}, Y^{t-1}}$ has conditional covariance $\Sigma_t$,  independent of $(A^{t-1}, Y^{t-1})$, satisfying   a filtering matrix  DRE  (see  (\ref{dre_1})) and implies, 
\begin{align}
& H(Y_t|Y^{t-1}, A^t)= H(I_t), \hso I_t\tri Y_t-{\bf E}\big\{Y_t|Y^{t-1}, A^t\big\}, \hso \forall t \in {\mathbb Z}_+^n, \nonumber  \\
&H(I_t)=\frac{1}{2} \log\big( (2\pi e)^{n_y} \det\big(C_t \Sigma_{t} C_t^T + N_t K_{W_t} N_t^T\big) \big).
\end{align}
{\it  2.2)  Theorem~\ref{thm_SS}.} An equivalent characterization of $C_{FB,n}(\kappa)$  corresponding to the  optimal strategy of  $A^n$ is realized  by 
\begin{align}
&A_t=  \Gamma_{t}^1\hat{X}_t + U(Y^{t-1}) +Z_t,  \: \: U(Y^{t-1})=\Gamma_t^2 \widehat{\hat{X}}_t, \:  \forall t  \in {\mathbb Z}_+^n,  \label{suf_s} \\
&\hat{X}_t \tri {\bf E}\big\{X_t\big|Y^{t-1}, A^{t-1}\big\}, \; \widehat{\hat{X}}_t\tri {\bf E}\big\{\hat{X}_t\big|Y^{t-1}\big\}, \; Z_t \in G(0,K_{Z_t}) \nonumber 
\end{align}
where $Z^n$ is an independent Gaussian process, independent of $(\hat{X}^n, \widehat{\hat{X}}^n)$. The above means,  for  each  $t$, $\mu_t^z(\cdot)$ depends on the {\it finite-dimensional sufficient statistic}  $\big(\hat{X}_t, \widehat{\hat{X}}_t, Z_t\big)$.  Moreover,  $C_{FB,n}(\kappa)$ is expressed,  as a functional of  two  matrix DREs of filtering theory of Gaussian systems,  $\Sigma_t$ and    $K_{t} \tri  {\bf E}\big\{\big(\hat{X}_t- \widehat{\hat{X}}_{t}\big)\big(\hat{X}_t-\widehat{\hat{X}}_{t} \big)^T\big|Y^{t-1}\big\}$, which is also independent of $Y^{t-1}$.\\
{\it 2.3) Theorem~\ref{thm_SS}.} A decentralized separation principle holds as follows:  \\
(i) The optimal strategy $U(\cdot)$ is explicitly determined from the solution of a stochastic LQG optimal control problem of minimizing the average cost \cite{caines1988}, and  involves the solution of a  control matrix DRE $P_t\succeq 0$, evolving  backward in time.   \\
(ii) The optimal strategy $(\Gamma^1, K_{Z})$ is determined by maximizing the   differential  entropy of $H(Y^n)$ subject to the average constraint with optimal  $U(\cdot)$ obtained in (i). \\
Finally, we  discuss  the limit  convergence,  $C_{FB} \tri \lim_{n \longrightarrow \infty} \frac{1}{n} C_{FB,n}$,  using the convergence  properties of  matrix DRE to analogous matrix algebraic Riccati equations (AREs),  as in \cite[Section~III]{charalambous2020new}.

\subsection{Relations to Past Literature}
$C_{FB}(\kappa)$ of  channels with memory is often characterized by directed information from the input of  the channel $A^n$ to the output of the channel $Y^n$, i.e., $I(A^n\rar Y^n)$  optimized over channel inputs with feedback,  $P_t(da_t|a^{t-1}, y^{t-1}), \forall t \in {\mathbb Z}_+^n$ subject to cost constraints \cite{massey1990,permuter-weissman-goldsmith2009}.
  Cover and Pombra \cite[Theorem~1]{cover-thomas2006} characterized  $C_{FB}(\kappa)$  of additive  Gaussian noise (AGN) channels, $Y_t=X_t+V_t, \forall t \in {\mathbb Z}_+^n, \frac{1}{n} {\bf E}\{\sum_{t=1}^n (A_t)^2\}\leq \kappa$,  with nonstationary and nonergodic jointly Gaussian noise $V^n$.
$C_{FB}(\kappa)$  of     Cover and Pombra AGN channel is characterized in  \cite{yang-kavcic-tatikonda2007,kim2008,kim2008,gattami2019}
for the degenerate case when the noise is    asymptotically stationary, with   state space realization of $V^n$.
  $V_t= C_t X_{t} +   W_t$, 
$X_{t+1}=F_{t} X_{t}+   W_t,\forall t$, such that  $P_t(da_t|a^{t-1}, y^{t-1})=P_t(da_t|x^{t}, y^{t-1})$,  i.e., the encoder  distribution at each time $t$,   knows the state of the channel $X^t$ 
(see the hidden assumptions in  \cite[Theorem~6.1, Lemma~6.1]{kim2010} and \cite{gattami2019}). These  assumptions are explicit in the analysis in \cite{yang-kavcic-tatikonda2007}.  When the state $X^n$ is not known to the encoder,  a sequential version of  the Cover and Pombra $C_{FB}(\kappa)$ is first derived in  \cite[Theorem~II.3]{charalambous2020new} and is different from \cite[Theorem~6.1, Lemma~6.1]{kim2010} and \cite{yang-kavcic-tatikonda2007,gattami2019}. The method of  \cite{charalambous2020new} is used in \cite{CDC-CK-SL:ITW2021}  to derive $C_{FB}(\kappa)$ for MIMO channels.
 For  finite-state channels with memory \cite{permuter-asnani-weissman2013,permuter-cuff-roy-weissman2010}, treated  $C_{FB}(\kappa)$, when the state of the channel is known to the encoder. \\ 
 The N-POSS (\ref{NDM-3})-(\ref{NDM-6})  includes the models of the above references.  More importantly, the code does not assume knowledge of the state $X^n$ at the encoder,  and $X^n$ is affected by $A^{n-1}$. This  generality requires additional  concepts  to characterize $C_{FB}(\kappa)$.  However, in  Remark~\ref{rem:rel}, we show that   from $C_{FB,n}(\kappa)$  of the LQG-POSS,  we  recover $C_{FB,n}(\kappa)$ of \cite{charalambous2020new,CDC-CK-SL:ITW2021} and \cite{yang-kavcic-tatikonda2007,kim2008,kim2008,gattami2019}
 as special  cases.

 \section{Characterization of $C_{FB,n}$ via Information State}
  \label{sect:AGN}
First, we   introduce the characterization  of $C_{FB,n}(\kappa)$ of   the N-POSS. 

\begin{theorem} ($C_{FB,n}$  with randomized strategies) \\
\label{thm_cons}
Consider the N-POSS,  the code $ \{(n, {\cal M}^{(n)}, \epsilon_n, \kappa)|n=1, 2, \dots\}$, and  assume the conditions of 
the converse/direct coding theorems  in  \cite{charalambous-kourtellaris-tziortzis:SICON-2024} hold. 
 Define the randomized strategies ${\cal P}_{n}(\kappa)$ 
by   
\begin{align*}
&{\cal P}_{n}(\kappa) \tri \Big\{ {\bf P}_{A_t|A^{t-1}, Y^{t-1}}={P}_t(da_t\big| a^{t-1},  y^{t-1}), \forall t \in {\mathbb Z}_+^n\Big| \nonumber  \\
 &\mbox{ {\bf  (C1)} holds and } \frac{1}{n} {\bf E}^{ P}\big( c_n(X^n,A^n)\big) \leq \kappa   \Big\} \subseteq {\cal E}_n(\kappa) 
\end{align*} 
where  (C1) follows from  the definition of the feedback code. 
The  directed  information from  $A^n$ to $Y^n$  is defined by 
\begin{align}
I^P(A^n \rar  Y^n)  \tri 
 {\bf E}^P \Big\{ \sum_{t=1}^n
\log\Big(\frac{{\bf P}_{Y_t|Y^{t-1}, A^t}^P}{{\bf P}^{P}_{Y_t|Y^{t-1}}}\Big) \Big\}  \label{di_general}
\end{align}
where  ${\bf P}_{Y_t|Y^{t-1}, A^t}^P $ and ${\bf P}_{Y_t|Y^{t-1}}^P $,  for $\forall t \in {\mathbb Z}_+^n$ are given by 
\begin{align}
&{\bf P}_{Y_t|Y^{t-1}, A^t}^P = \int_{{\mathbb X}} Q_t(dy_t|x_t, a_t)
  {\bf P}_t^P(dx_t|y^{t-1}, a^{t-1}),   \label{cd_c1}  \\
&{\bf P}_{Y_t|Y^{t-1}}^P = \int_{{\mathbb A}^t  \times {\mathbb X}^t}   { Q}_t(dy_t\big| x_t, a_t)  { P}_t(da_t|y^{t-1}, a^{t-1}) \nonumber \\
&\hspace{2.0cm}  {\bf P}_t^P(dx_t\big|a^{t-1}, y^{t-1})
 {\bf P}_t^P(da^{t-1}|y^{t-1}).   \label{cd_o1}
\end{align}
The CC capacity  is given  by  ${C}_{FB}(\kappa) \tri \lim_{n \longrightarrow \infty} \frac{1}{n} {C}_{FB,n}(\kappa)$, 
where   
 ${C}_{FB,n}(\kappa)\tri  \sup_{{\cal P}_{n}(\kappa)}  I^P(A^n \rar Y^n)$.  
\end{theorem}
\begin{proof} The derivation is based on  of  the proof of   the converse/direct coding theorems  in  \cite{charalambous-kourtellaris-tziortzis:SICON-2024}, hence we omit it. 
\end{proof}

In  $C_{FB,n}(\kappa)$,
 the fundamental element that determines the properties of  the maximizing elements  $\{{ P}_{t}(da_t|a^{t-1}, y^{t-1})|    t\in {\mathbb Z}_+^n\}$ is the a posteriori distribution  $\{{\bf P}_{t}^P(dx_t|a^{t-1}, y^{t-1})| t\in {\mathbb Z}_+^n\}$, and its relation to an informations state, as defined in Appendix~\ref{app:is-suf}.

 \subsection{Information States and  Sufficient Statistics}
Now we prepare to show the statements of Section~\ref{sect:pr-mr}. 

\begin{lemma} 
\label{lemma_c}
  (Consistent family of randomized strategies) \\
(a) For any element of ${\cal P}_{n}(\kappa)$ there exist independent  uniformly distributed 
 RVs $U^n$,  $U_t : \Omega \rar {\mathbb U}_t=[0,1], \forall t\in {\mathbb Z}_+^n$, 
 and  measurable functions  ${\mu}_t^u(\cdot)$ of  $\{a^{t-1},y^{t-1},u_t\}$, such that
\begin{align}
 & a_t={\mu}_t^u(a^{t-1},y^{t-1}, u_t), \hso \forall t \in {\mathbb Z}_+^n,  \\ 
&{ P}_t(da_t| a^{t-1},y^{t-1}) = {\bf P}_{U_t} \big( u_t \in  [0,1]\big|{\mu}_t^u(a^{t-1}, y^{t-1}, u_t)  \in da_t \big), 
\nonumber 
\\
&\mbox{{\bf (C2).} $U^n$  independent RVs,  $U_t$ independent of   $(W^{t}, X_1)$,} \nonumber \\
&\hst \mbox{   ${\bf P}_{A_t|A^{t-1}, Y^{t-1}, W^t, X^t}={ P}_t(da_t| a^{t-1},y^{t-1}) , \forall t\in {\mathbb Z}_+^n$}. \nonumber 
 \end{align}
Moreover, $\delta_t^u\tri (a^{t-1}, y^{t-1}, u_t)$ is a sufficient statistic for ${ P}_t(da_t| a^{t-1},y^{t-1})$. \\ 
(b) Similarly to (a)  there exist arbitrary independent  
 RVs $Z^n$,  $Z_t : \Omega \rar {\mathbb Z}_t, \forall t\in {\mathbb Z}_+^n$, and  meas. functions  ${\mu}_i^z(\cdot)$  such that
\begin{align}
 & a_t={\mu}_t^z(a^{t-1},y^{t-1}, z_t),  \hso \forall t \in {\mathbb Z}_+^n,
 \label{AR_IL}\\
&{ P}_t(da_t| a^{t-1},y^{t-1}) = {\bf P}_{Z_t} \big( z_t \in  {\mathbb Z}_t \big|{\mu}_t^z(a^{t-1}, y^{t-1}, z_t)  \in da_t \big) \nonumber 
\\
&\mbox{{\bf (C3).} $Z^n$  independent RVs,  $Z_t$ independent of   $(W^{t}, X_1)$,} \nonumber \\
&\hst \mbox{   ${\bf P}_{A_t|A^{t-1}, Y^{t-1}, W^t, X^t}={ P}_t(da_t| a^{t-1},y^{t-1}) , \forall t\in {\mathbb Z}_+^n$}.\nonumber 
 \end{align}
Moreover, $\delta_t^z\tri (a^{t-1}, y^{t-1}, z_t)$ is a sufficient statistic for ${ P}_t(da_t| a^{t-1},y^{t-1})$.
\end{lemma}
\begin{proof} (a) The existence of $\mu_t^u$ with the stated properties  is an application of 
 \cite[Lemma~1.2]{gihman-skorohod1979}, as described in  \cite[Section~III]{charalambous-kourtellaris-loykaIT2015_Part_2};  (C2) is  inherited from condition (C1), and $\delta_t^u$ is due to the definition of  
  Definition~\ref{def:is}.(b). 
    (b) Follows from properties of quantile representation of distributions  \cite[Section~III]{charalambous-kourtellaris-loykaIT2015_Part_2}.
\end{proof}

Next, we verify  ${\bf P}_{t}^P(dx_t|a^{t-1}, y^{t-1})$ is an information state using  strategies $\mu_t^u(\cdot), \mu_t^z(\cdot)$, $\forall t\in {\mathbb Z}_+^n$ of   Lemma~\ref{lemma_c}.

\begin{theorem} (Information state and separated strategies)\\
\label{theorem:aposteriori}
Consider the N-POSS and $C_{n,FB}(\kappa)$ of Theorem~\ref{thm_cons}.\\
(a) The a posteriori distribution ${\bf P}^P_{t+1}(dx_{t+1}|a^{t}, y^{t})$ satisfies the recursion\footnote{All distributions are assumed to have probability densities i.e., ${\bf P}^P_{t+1}(dx_{t+1}|a^{t}, y^{t})={\bf p}^P_{t+1}(x_{t+1}|a^{t}, y^{t})dx_{t+1}$.}, $\forall t \in {\mathbb Z}_+^n$: 
 \begin{align}
&{\bf P}_{t+1}^P(dx_{t+1}\big|a^{t}, y^{t}) \label{apost_1}    \\
&=\frac{{\bf T}_{t+1}(y_{i}, a_{t}, {\bf P}_{t}^P(\cdot\big|a^{t-1}, y^{t-1}))(dx_{t+1})    }{\int_{ {\mathbb X}}  {\bf T}_{t+1}(y_{t}, a_{t}, {\bf P}_{t}^P(\cdot\big|a^{t-1}, y^{t-1}))(dx_{t+1})    },  \: {\bf P}_1^P(dx_1\big|y^{0})= {\bf P}_{X_1} \nonumber \\
&{\bf T}_{t+1}(y_{t}, a_{t},{\bf P}_{t}^P(\cdot\big|a^{t-1}, y^{t-1}))(dx_{t+1}) \tri\int_{{\mathbb X}} S_{t+1}(dx_{t+1}\big|y_{t}, x_{t},a_{t}) \nonumber \\
&\hst \otimes Q_{t}(dy_{t}\big|x_{t}, a_{t})   \otimes {\bf P}_{t}^P(dx_{t}\big|a^{t-1}, y^{t-1}) .   
\end{align}
and $\{{\bf P}^P_{t+1}(dx_{t+1}|a^{t}, y^{t})\big| t \in {\mathbb Z}_+^n\}$ is an information state. \\
(b)  $\forall t \in {\mathbb Z}_+^n$, distr. ${ P}_t(da_t| a^{t-1},y^{t-1})$ in the definition of distr.  ${\bf P}_{Y_t|Y^{t-1}}^P$, i.e., (\ref{cd_o1}), are induced by $\mu^u(\cdot)$ or $\mu^z(\cdot)$ of Lemma~\ref{lemma_c}.(a), (b).\\
(c)  $C_{FB,n}(\kappa)$ is  re-formulated  using     separated strategies,  $a_t=\mu_t^u(\xi_t(a^{t-1}, y^{t-1}),u_t)$  or  $a_t=\mu_t^z(\xi_t(a^{t-1}, y^{t-1}),z_t)$, $\xi_t(a^{t-1}, y^{t-1})\tri {\bf P}_{t}^P(dx_t|a^{t-1}, y^{t-1})$  as follows.  
\begin{align}
&C_{FB,n}(\kappa) \\
& =\sup_{ \{\mu_t^u(\xi_t(a^{t-1}, y^{t-1}),u_t)   \}_ {t=1}^n, \;  \frac{1}{n}{\bf E}^{\mu^u}\big( c_n(A^n,X^n) \big) \leq \kappa}\sum_{t=1}^n I^{\mu^u}(A^t; Y_t|Y^{t-1})\nonumber \\
&= \sup_{ \{ \mu_t^z(\xi_t(a^{t-1}, y^{t-1}),z_t)   \}_ {t=1}^n, \;  \frac{1}{n}{\bf E}^{\mu^z}\big( c_n(A^n,X^n) \big) \leq \kappa}  \sum_{t=1}^n I^{\mu^z}(A^t; Y_t|Y^{t-1}). \nonumber
\end{align}
%
\end{theorem}
\begin{proof} (a) By the existence of probability densities, 
${\bf P}_{t+1}^P(dx_{t+1}\big|a^{t}, y^{t})={\bf p}_{t+1}^P(x_{t+1}\big|a^{t}, y^{t})dx_{t+1}$,   ${\bf p}_{t+1}^P(x_{t+1}\big|a^{t}, y^{t})=\frac{{\bf p}_{t+1}^P(x_{t+1}, a^{t}, y^{t})}{{\bf p}_{t}^P(a^{t}, y^{t})}$. Then ${\bf p}_{t+1}^P(x_{t+1}, a^{t}, y^{t})= \int_{{\mathbb X}_{t} }{\bf p}_{t+1}^P(x_{t+1}, x_t, a^{t}, y^{t})$. Using the N-POSS distributions and condition (C1) we obtain the recursion.   By  Definition~\ref{def:is} 
the  a posteriori distribution   is an information state. (b) Follows  from Lemma~\ref{lemma_c}. (c) Follows from  (a), (b).  
\end{proof}

\vspace*{-0.3cm}

Next, 
we restrict $\mu^{u}(\cdot), \mu^{z}(\cdot)$ to be information lossless  to show  $I^P(A^n\rar Y^n)=\sum_{t=1}^n I^{\mu^u}(U^t; Y_t|Y^{t-1})=\sum_{t=1}^n I^{\mu^z}(Z^t; Y_t|Y^{t-1})$.

\begin{theorem}(Equivalent characterizations of $C_{FB,n}(\kappa)$)\\
\label{thm:IL}
Consider the N-POSS and  $C_{FB,n}(\kappa)$ of Theorem~\ref{thm_cons}. Define  the set of  information lossless strategies,
\begin{align}
&{\cal P}_{n}^{IL,\mu^u}(\kappa)\tri \big\{a_t=\mu_t^u(a^{t-1}, y^{t-1}, u_t), \forall t \in {\mathbb Z}_+^n\big|\mbox{the maps } \nonumber \\
&\mbox{$ \mu_t^u(a^{t-1}, y^{t-1}, \cdot): {\mathbb U}_t \rar a_t, \forall {\mathbb Z}_+^n$ are  bijections, with meas.} \nonumber  \\
&\mbox{inverses,   {\bf (C2)} holds,}  \; {\bf E}^{\mu^u}\{c_n(A^n, X^n)\}\leq \kappa n  \Big\}.
\label{il}
\end{align}
Similarly, define ${\cal P}_{n}^{IL,\mu^z}(\kappa)$ with $(u_t,{\mathbb U}_t)$ replaced by $(z_t,{\mathbb Z}_t)$. Then  
\begin{align}
C_{FB,n}(\kappa)=& \sup_{{\cal P}_{n}^{IL,\mu^u}(\kappa)} \sum_{t=1}^n I(U^t; Y_t|Y^{t-1})\label{di-un_a} \\
=&  \sup_{\{{\bf P}_{Z_t}\}_{t=1}^n, \: {\cal P}_{n}^{IL,\mu^z}(\kappa)} \sum_{t=1}^n I(Z^t; Y_t|Y^{t-1}). \label{di-un_b}
\end{align}
Moreover, in (\ref{di-un_a}) and (\ref{di-un_b}) ${\cal P}_{n}^{IL,\mu^u}(\kappa)$ and ${\cal P}_{n}^{IL,\mu^z}(\kappa)$ can be replaced by separated strategies (see Definition~\ref{def:is}). 
\end{theorem}
\vspace*{-0.3cm}
\begin{proof} (\ref{di-un_a}), (\ref{di-un_b})  follow from the bijection property of the information lossless strategies and Theorem~\ref{theorem:aposteriori}. 
\end{proof}

\vspace*{-0.3cm}
In Section~\ref{appl} we use   information states, sufficient statistics,  and separated strategies to  characterize $C_{FB,n}(\kappa)$ of the LQG-POSS.

\section{Characterization of Feedback Capacity of LQG-POSS via  Information State}
\label{appl}
We characterize $C_{FB,n}(\kappa)$ of the  LQG-POSS (\ref{LQG-1})-(\ref{LQG-7}) in several steps, using  Section~\ref{sect:AGN}, and a generalization of  the   sufficient statistic approach  in  
\cite{charalambous2020new,CDC-CK-SL:ITW2021}.\\
{\it  Step \#1. Gaussian $(A^n, X^n, Y^n)$ are optimal. }  We show the statement of  Section~\ref{sub-section:A.2}, under 2.1). 

\begin{lemma}  (Calculation of $H^P(Y_t|Y^{t-1}, A^t)$)\\
\label{lemma_POSS}
Consider  the LQG-DM  (\ref{LQG-1})-(\ref{LQG-7}). Define  $\forall t \in {\mathbb Z}_+^{n}$,  
\begin{align}
&\hat{I}_t \tri Y_t -{\bf E}\big\{Y_t\big|Y^{t-1}, A^t\big\}, \: \mbox{innovations  wrt $Y^{t-1}, A^t$}\\
&\mu_{Y_t|Y^{t-1}, A^t} \tri {\bf E}\big\{Y_t\big|Y^{t-1}, A^{t}\big\}, \: K_{Y_t|Y^{t-1}, A^t} \tri  cov\big(Y_t, Y_t\big|Y^{t-1}, A^{t}\big),\nonumber  \\
&\hat{X}_t \tri {\bf E}\big\{X_t\big|A^{t-1}, Y^{t-1}\big\}, \: \Sigma_t \tri  cov\big(X_t, X_t\big|A^{t-1}, Y^{t-1}\big).\nonumber  
\end{align}
(a) 
The solution of  the recursion of Theorem~\ref{theorem:aposteriori}.(a)   is conditionally Gaussian given by 
\begin{align}
{\bf P}^P_{t}(dx_{t}|A^{t-1}, Y^{t-1})\in G(\hat{X}_{t};\Sigma_{t}), \ \ \forall t \in {\mathbb Z}_+^n
\end{align}
where $\hat{X}_{t}$ is linear in $(A^{n-1}, Y^{t-1})$,  $\Sigma_{t}={\bf E}\big\{(X_t-\hat{X}_t)(X_t-\hat{X}_t)^T\big\}$ is  independent of $(A^{t-1}, Y^{t-1})$,  and   satisfy   recursions:\\
(i)  $\hat{X}_t$ satisfies the generalized Kalman-filter recursion,   
\begin{align}
&\hat{X}_{t+1}=F_{t} \hat{X}_{t}+B_t A_t+  M_{t}(\Sigma_t)  \hat{I}_t,  \: \hat{X}_{1}=\mu_{X_1}, \: \forall  t\in {\mathbb Z}_+^{n-1},\label{kal_fil_noise} \\
& M_{t}(\Sigma_t) \tri \Big( F_{t}  \Sigma_{t} C_{t}^T+G_{t} K_{W_{t}}N_{t}^T\Big)\Big(N_{t} K_{W_{t}}N_{t}^T+ C_{t} \Sigma_{t} C_{t}^T \Big)^{-1},\nonumber  \\
&\hat{I}_t =Y_t - C_t \hat{X}_{t}-D_t A_t= C_t\big(X_{t}- \hat{X}_{t}\big) + N_t W_t,  \label{inn_po_1} \\
& \hat{I}_t \in G(0, K_{\hat{I}_t}) \hso  \mbox{orth.  process  indep. of  $Y^{t-1}, A^t$},  \label{inn_po_2}\\
&K_{\hat{I}_t} \tri  cov(\hat{I}_t, \hat{I}_t)= C_t \Sigma_t C_t^T +N_t K_{W_t} N_t^T=K_{Y_t|Y^{t-1}, A^t},  \label{cov_in_noise} \\
& \mu_{Y_t|Y^{t-1}, A^t}=C_t \hat{X}_{t}+D_t A_t.
\end{align}
(ii)  $\Sigma_t $ satisfies the matrix DRE,   
\begin{align}
&\Sigma_{t+1}= F_{t} \Sigma_{t}F_{t}^T  + G_{t}K_{W_{t}}G_{t}^T -\Big(F_{t}  \Sigma_{t}C_{t}^T+G_{t}K_{W_{t}}N_{t}^T  \Big) \nonumber \\
& \hspace{0.1cm} . \Big(N_{t} K_{W_{t}} N_{t}^T+C_{t}  \Sigma_{t} C_{t}^T\Big)^{-1}\Big( F_{t}  \Sigma_{t}C_{t}^T+ G_{t} K_{W_{t}}N_{t}^T  \Big)^T,\nonumber\\   & \hspace{0.1cm}\hso \Sigma_t \succeq 0, \hso  \forall t\in {\mathbb Z}_+^n, \hso \Sigma_{1}=K_{X_1}\succeq 0. \label{dre_1}
\end{align}
(b) The conditional entropy   $H^P(Y_t|Y^{t-1}, A^t)$ is   given by 
\begin{align}
&H^P(Y_t|Y^{t-1}, A^t)=H(\hat{I}_t), \hso  \: \forall t \in {\mathbb Z}_+^n \\
&= \frac{1}{2} \log\big( (2\pi e)^{n_y} \det\big(C_t \Sigma_{t} C_t^T + N_t K_{W_t} N_t^T\big) \big). \label{entr_noise}
\end{align}
\end{lemma}
\begin{proof} All statements follow by verifying that  ${\bf P}^P_{t}(dx_{t}|A^{t-1}, Y^{t-1})=\frac{1}{(2\pi)^{n_x/2}(\det(\Sigma_t))^{1/2}}\exp\big(-\frac{1}{2}(x_t- \hat{X}_t)\Sigma_t^{-1}(x_t- \hat{X}_t)^T\big), \forall t\in {\mathbb Z}_+^n$   satisfies recursion (\ref{apost_1}).
%
\end{proof}


\vspace*{-0.2cm}

\begin{lemma} (Optimality of Gaussian inputs)\\
\label{lemma_Gaussian} 
Consider  the LQG-POSS  (\ref{LQG-1})-(\ref{LQG-7}). 
The optimal  $P(\cdot|\cdot)  \in  {\cal P}_n(\kappa)$ which maximizes $I(A^n\rar Y^{n}) $  is induced by the realization of  jointly Gaussian  $A^n$,  given by 
\begin{align}
&A_t= \mu_t^z(Y^{t-1},A^{t-1},  Z_t)=\mu_t^0(A^{t-1}, Y^{t-1}) +Z_t, \label{rec_2_n} \\
&\mu_t^0(A^{ t-1}, Y^{t-1}) \hso \mbox{is linear in the RVs $(A^{t-1},  Y^{t-1})$}, \\
&Z_t \in G(0, K_{Z_{t}}), \; K_{Z_{t}}\succeq 0, \forall t\in {\mathbb Z}_+^n  \hso \mbox{indep. proc.} \\
&Z_i \hso \mbox{indep. of} \hso (X^t, W^t,A^{t-1}, Y^{t-1}), \; \forall t\in {\mathbb Z}_+^n.  \label{rec_2_nn} 
\end{align}
\end{lemma}
\vspace*{-0.6cm}
\begin{proof} Follows from Lemma~\ref{lemma_POSS}.(b), by the maximum entropy principle, because  $I(A^n\rar Y^n) =H(Y^n)-\sum_{t=1}^n H(\hat{I}_t)$, and 
 the average cost constraint is quadratic.
By Lemma~\ref{lemma_c}.(b),  one  realization of $A^n$ is (\ref{rec_2_n})-(\ref{rec_2_nn}). 
\end{proof}

\vspace*{-0.3cm}

{\it Step \#2. Equivalent Characterization  of  $ C_{FB,n}(\kappa)$ via  Information State and Sufficient Statistics.}  
First, we invoke Lemma~\ref{lemma_POSS},  to identify a sufficient statistic for ${\cal P}_n(\kappa)$.

\begin{lemma} Preliminary  char. of  ${C}_{FB,n}(\kappa)$ for LQG-POSSs  \\ 
\label{thm_SS_P}
Consider  the LQG-POSS  (\ref{LQG-1})-(\ref{LQG-7}). 
 Then
\begin{align}
&Y_t= C_t{\hat{X}}_{t} +D_t A_t +\hat{I}_t, \hso  t=1, \ldots, n,  \label{cp_16_alt_n_P}\\
& {\bf P}_t(dy_t|y^{t-1},a^t)= {\bf P}_t(dy_t|a_t, \hat{x}_t),  \label{PO_tr_1_P_a} \\
& {\bf P}_t(dy_t|y^{t-1})=\int {\bf P}_t(dy |a_t,\hat{x}_{t}){\bf P}_t(da_t|\hat{x}_t,y^{t-1})\nonumber \\
&\hso . {\bf P}_t(d\hat{x}_t|y^{t-1}), 
 \label{PO_tr_1_P}\\
&P_t(da_t|a^{t-1}, y^{t-1})={\bf P}_t(da_t|\hat{x}_{t}, y^{t-1}) \label{PO_tr_2_b_P} \\
&A_t=\Gamma_t^1 \hat{X}_t + U_t(Y^{t-1}) + Z_t,\hso  U_t(Y^{t-1})\tri  \Gamma_{t}^2 Y^{t-1},  \label{cp_16_alt_n_P_aa_n} \\
&  Z_t \in G(0, K_{Z_t}) \hso \mbox{indep. of} \hso (X^{t}, \hat{X}^t, \hat{I}^t, A^{t-1}, Y^{t-1}). 
\end{align}
where $(\Gamma_t^1, \Gamma_t^2)$ are  nonrandom. 
An equivalent 
${C}_{FB,n}(\kappa)$     is 
\begin{align}
&{C}_{FB, n}(\kappa)= \sup_{{\cal P}_{n}^{\hat{X}}(\kappa)}\sum_{t=1}^n I(A_t, \hat{X}_t; Y_t|Y^{t-1}), \label{ftfic_is_P}  \\
&{\cal P}_{n}^{\hat{X}}(\kappa) \tri \Big\{ \{{\bf P}_t(da_t|\hat{x}_{t}, y^{t-1})\}_{t=1}^n\Big|  \frac{1}{n} {\bf E}\big( \sum_{t=1}^n \hat{\gamma}_t(A_t, \hat{X}_t)\big) \leq \kappa    \Big\} \nonumber\\
& \hat{\gamma}_t(A_t, \hat{X}_t) \tri \langle A_t, R_{t} A_t \rangle+ \langle \hat{X}_t, Q_{t} \hat{X}_t \rangle + trace\big(Q_t \Sigma_t\big).   
\label{PO_tr_3_P_a}  
\end{align}
\end{lemma}
\vspace*{-0.3cm}
\begin{proof} See Section~\ref{app:thm_SS_P}. 
\end{proof}

\vspace*{-0.4cm}
{\it Step \#3. Complete Characterization of  $C_{FB,n}(\kappa)$. } Now, we show  the statements  of Section~\ref{sub-section:A.2},  under  2.2).  We use Lemma~\ref{thm_SS_P} to  find another equivalent $Y^n$  to compute ${\bf P}_{Y_t|Y^{t-1}}$, which is required in the optimization of  $H(Y^n)$. 

%

\begin{theorem} Complete char.  of  ${C}_{FB,n}(\kappa)$ for LQG-POSSs   \\ 
\label{thm_SS}
Consider  the LQG-POSS  (\ref{LQG-1})-(\ref{LQG-7}).  
Define  
\begin{align}
I_t\tri & Y_t - {\bf E}\big\{Y_t\big|Y^{t-1}\big\},\: \mbox{$Y^{t}$ generated by (\ref{cp_16_alt_n_P})}, \\ 
\widehat{\hat{X}}_{t} \tri& {\bf E}\big\{\hat{X}_t\Big|Y^{t-1}\big\},  \; \forall t \in {\mathbb Z}_+^n, \; 
\widehat{\hat{X}}_1 \tri \mu_{X_1}, \; K_1 \tri 0,  \nonumber 
\\
K_{t} \tri & cov\big(\hat{X}_t,\hat{X}_t\Big|Y^{t-1}\big)= {\bf E}\big\{\big(\hat{X}_t- \widehat{\hat{X}}_{t}\big)\big(\hat{X}_t-\widehat{\hat{X}}_{t} \big)^T\big|Y^{t-1}\big\}.\nonumber 
\end{align}
(a) An  equivalent  representation of $(A^n, Y^n)$ is
\begin{align}
&Y_t =\big(C_t+ D_t \Gamma_t^1\big) \widehat{\hat{X}}_t +D_t U_t(Y^{t-1})  + {I}_t, \hso \forall t \in {\mathbb Z}_+^n,    \label{cp_16_alt_inn}  \\
&A_t=\Gamma_t^1 \hat{X}_t + U_t(Y^{t-1}) + Z_t, \hso    U_t(Y^{t-1})\tri  \Gamma_{t}^2 Y^{t-1},    \label{cp_16_alt_n_P_aa} \\
&  Z_t \in G(0, K_{Z_t}) \hso \mbox{ind. of} \hso (X^{t}, \hat{X}^t, \widehat{\hat{X}}_t,  \hat{I}^t, A^{t-1}, Y^{t-1}) 
\end{align}
where $\widehat{\hat{X}}_t$ and $K_t$ are given under (i) and (ii) below. \\
(i)   $\widehat{\hat{X}}_t$ satisfies the Kalman-filter recursion,
\begin{align}
&\widehat{\hat{X}}_{t+1}=F_{t}(\Gamma_t^1) \widehat{\hat{X}}_{t}+ B_t U_t\nonumber \\
&+F_{t}^{CL}(\Sigma_{t}, K_t, \Gamma_t^1)  I_t, \; \widehat{\hat{X}}_{1}=\mu_{X_1},\hso \forall t \in {\mathbb Z}_+^{n-1}, \label{kf_m_1} \\
& F_t(\Gamma_t^1)\tri F_t +B_t \Gamma_t^1,    \;  C_t(\Gamma_t^1) \tri C_t + D_t\Gamma_t^1 ,   \\
&F_{t}^{CL}(\Sigma_{t}, K_t, \Gamma_t^1) \tri \Big(F_{t}(\Gamma_t^1)  K_{t}\big(C_t(\Gamma_t^1) \big)^T+ B_t K_{Z_t}D^T  \\
&+  M_{t}(\Sigma_{t}) K_{\hat{I}_{t}} \Big)\big\{K_{\hat{I}_t}+ D_t  K_{Z_t}D_t^T + \big(C_t (\Gamma_t^1) \big) K_{t} \big(C_t(\Gamma_t^1) \big)^T \big\}^{-1} ,  \nonumber \\
&I_t = Y_t-C_t\widehat{\hat{X}}_{t}- D_t \Gamma_t^1 \widehat{\hat{X}}_{t}-D_t U(Y^{t-1})     \nonumber\\ &=  C_t(\Gamma_t^1) \Big(\hat{X}_t-\widehat{\hat{X}}_{t}\Big)+ \hat{I}_t+ D_t Z_t, \hso \forall t \in {\mathbb Z}_+^n, \label{kf_m_2} \\
&I_t \in  G(0; K_{I_t}), \; \forall t \in {\mathbb Z}_+^n \; \mbox{orthogonal innovations proc.}\nonumber \\
&K_{I_t}=K_{Y_t|Y^{t-1}} \tri cov\big(I_t,I_t\big)=  \Big(C_t(\Gamma_t^1)\Big)K_t \Big(C_t(\Gamma_t^1)\Big)^T \nonumber \\
&+ K_{\hat{I}_t} + D_tK_{Z_t}D_t^T,  \hso
K_{\hat{I}_t}  \hso  \mbox{given by  (\ref{cov_in_noise})}. \label{inno_PO}
\end{align}
(ii)  $K_t = {\bf E}\big\{\widehat{E}_t \widehat{E}_t^T\big\}$, $\widehat{E}_t = \hat{X}_{t}- \widehat{\hat{X}}_{t}$ satisfies the matrix DRE 
\begin{align}
&K_{t+1}= F_t(\Gamma_t^1) K_{t}\big(F_t(\Gamma_t^1\big)^T  + M_t(\Sigma_{t})K_{\hat{I}_t}\big(M_t(\Sigma_{t})\big)^T  +B_t K_{Z_t} B_t^T    \nonumber  \\
& -\Big(F_t(\Gamma_t^1)  K_{t}\big(C_t (\Gamma_t^1) \big)^T   +B_t K_{Z_t} D_t^T  + M_t(\Sigma_t)K_{\hat{I}_t}   \Big) \nonumber \\
&.  \Big( K_{\hat{I}_t}+D_t K_{Z_t} D_t^T
+ \big(C_t (\Gamma_t^1 \big) K_{t} \big(C_t( \Gamma_t^1) \big)^T     \Big)^{-1}  \Big(  F_t(\Gamma_t^1)  K_{t}\big(C_t (\Gamma_t^1)  \big)^T\nonumber\\ 
&+ B_t K_{Z_t}B_t^T+ M_t(\Sigma_t)K_{\hat{I}_t}      \Big)^T, \: K_t \succeq 0, 
 \forall t \in {\mathbb Z}_+^n, \; K_1=0. \label{kf_m_4_a}  
\end{align} 
(b) An equivalent characterization of ${C}_{FB,n}(\kappa)$     is 
\begin{align}
&{C}_{FB,n}(\kappa)= \sup_{{\cal P}_{n}^{\widehat{\hat{X}}}(\kappa)}\sum_{t=1}^n I(A_t, \hat{X}_t; Y_t|Y^{t-1}), \label{ftfic_is}  \\
&I(A_t, \hat{X}_t; Y_t|Y^{t-1})= \frac{1}{2}  \log \frac{ \det( K_{I_t})   }{\det(K_{\hat{I}_t})}= \label{ftfic_is_aa}    \\
&\frac{1}{2} \log \Big\{ \frac{ \det \Big( C_t( \Gamma_t^1)K_t (C_t (\Gamma_t^1)^T + K_{\hat{I}_t} + D_tK_{Z_t}D_t^T \Big)    }{\det\big(K_{\hat{I}_t}\big)}\Big\},   \label{ftfic_is_a}     
\end{align}
\begin{align}
&{\cal P}_{n}^{\widehat{\hat{X}}}(\kappa) \tri \Big\{ \{ U_t, \Gamma_t^1, K_{Z_t} \}_{t=1}^n\Big|  \frac{1}{n} {\bf E}\big( \sum_{t=1}^n \widehat{\hat{\gamma}}_t(U_t,\widehat{\hat{X}}_t, \Gamma_t^1, K_{Z_t})\big) \leq \kappa    \Big\} \nonumber\\
& \widehat{\hat{\gamma}}_t(U_t,\widehat{\hat{X}}_t, \Gamma_t^1 K_{Z_t}, ) \tri \langle U_t, R_{t} U_t \rangle+ \langle \widehat{\hat{X}}_t, Q_{t}(\Gamma_t^1) \widehat{\hat{X}}_t \rangle +\langle \widehat{\hat{X}}_t, L_t(\Gamma_t^1)U_t \rangle   \nonumber \\
& +\langle L_t(\Gamma_t^1)U_t, \widehat{\hat{X}}_t \rangle   
+   trace \Big(Q_t \Sigma_t+ Q_t(\Gamma_t^1) K_t + R_t K_{Z_t}\Big),\label{PO_tr_3}  \\
& Q_t(\Gamma_t^1)\tri Q_t+ (\Gamma_t^1)^T R_t \Gamma_t^1, \hso L_t(\Gamma_t^1)\tri (\Gamma_t^1)^T R_t. 
\label{PO_tr_3_a}  
\end{align}
\end{theorem}
\vspace*{-0.cm}
\begin{proof}
(a) By definition of $I_t$ with $Y^t$ generated by  (\ref{cp_16_alt_n_P}), we obtain (\ref{cp_16_alt_inn}). The rest of equations under (i) and (ii) follow by repeating the derivation of Kalman-filter equations.
 (b) (\ref{ftfic_is_a}) follows from    $ I(A_t, \hat{X}_t; Y_t|Y^{t-1})=H(Y_t|Y^{t-1})- H(\hat{I}_t)$ from (\ref{entr_noise}),  and $H(Y_t|Y^{t-1})$ is computed from (\ref{inno_PO}). 
  ${\cal P}_{n}^{\widehat{\hat{X}}}(\kappa)$ follows from  Lemma~\ref{thm_SS_P}, (\ref{PO_tr_3_P_a}) by  ${\bf E}\{ \hat{\gamma}_t(A_t, \hat{X}_t)\} ={\bf E}\{{\bf E}\{ \hat{\gamma}_t(A_t, \hat{X}_t) \big|Y^{t-1}\}    \}$.   
\end{proof}

\vspace*{-0.3cm}
{ \it \#4. Separation Principle.} Now, 
we show the statement of   Section~\ref{sub-section:A.2}, under  2.3):    the strategy $U_t(Y^{t-1})=\Gamma_t^2 Y^{t-1}, $ is determined independently of  $(\Gamma_t^1\hat{X}_t, K_{Z_t}), \forall t \in {\mathbb Z}_+^n$.

\begin{theorem}(Decentralized separation principle)\\
\label{gen_exa}
Consider  $C_{FB,n}(\kappa)$  of Theorem~\ref{thm_SS}.(b). Define  the Cost-Rate, i.e.,  dual of $C_{FB,n}(\kappa)$, by 
 \begin{align}
&\kappa_{n}(C)\tri   \inf_{\{ \big(\Gamma_t^1, U_{t}, K_{Z_t}\succeq 0\big)\}_{t=1}^n  }{\bf E}\Big\{\sum_{t=1}^n \widehat{\hat{\gamma}}_t(U_t,\widehat{\hat{X}}_t, \Gamma_{t}^1,K_{Z_t})  \Big\} \nonumber \\
&\mbox{such that (\ref{kf_m_1}) holds, } \hso  \frac{1}{2} \sum_{t=1}^n    \log \frac{ \det( K_{I_t})   }{\det(K_{\hat{I}_t})}   \geq nC. \label{cap_fb_1_TC_1n_NN}
\end{align} 
{\it (a) Decentralized Separation  Principle.} \\
(i) 
The optimal strategy $\{U_t^{*}(\cdot,\Gamma^1, K_Z)| t \in {\mathbb Z}_+^n\}$ is  a solution of the stochastic  optimal control  problem 
\begin{align}
J_{n}^{SC}(U^*, \Gamma^1, K_Z) \tri \inf_{U(\cdot)}  {\bf E}\Big\{\sum_{t=1}^n  \widehat{\hat{\gamma}}_t(U_t,\widehat{\hat{X}}_t, \Gamma_{t}^1,K_{Z_t})  \Big\}  \label{sc_pr}
\end{align}
and it is given by the following equations. 
\begin{align}
&U_t^{*}(Y^{t-1},  \Gamma^1, K_{Z})={\Gamma}_{t}^{2,*}  \widehat{\hat{X}}_t, \hso \forall t\in {\mathbb Z}_+^{n-1}, \label{opt_con_1_1_n_a}\\
&{\Gamma}_{t}^{2,*} 
=-\Big(R_{t}+  B_{t}^T P_{t+1}B_t \Big)^{-1}\Big(  \big({L}_{t}(\Gamma_t^1)\big)^T 
\nonumber \\
&+ B_{t}^T P_{t+1} {F}_{t}(\Gamma_t^1)\Big), \hso  
 {\Gamma}_{n}^{2,*}=-R_n^{-1} \big({L}_{n}(\Gamma_n^1)\big)^T\label{opt_con_1_1_n}
\end{align}
where $P_t\succeq 0$  satisfies   the backward  matrix DRE, 
\begin{align}
&P_t = \big({F}_t(\Gamma_t^1)\big)^T  P_{t+1}{F}_t(\Gamma_t^1)  - \Big(\big({F}_t(\Gamma_t^1)\big)^T  P_{t+1} B_t  +{L}_{t}(\Gamma_t^1) \Big) \nonumber \\
& .\Big( R_t+ B_t^T P_{t+1} B_t \Big)^{-1} \Big( \big(F_t(\Gamma_t^1)\big)^T P_{t+1}B_t  \nonumber \\
&+ L_t(\Gamma_t^1) \Big)^T  +  Q_t(\Gamma_t^1),\: \forall t \in {\mathbb Z}_+^{n-1}, \hso P_n= Q_n(\Gamma_n^1).  \label{ric2}
\end{align}
The optimal cost-rate is given by
\begin{align}
&J_{n}^{SC}(U^*, \Gamma^1, K_Z)= 
 \langle \mu_{X_1} , P_1 \mu_{X_1} \rangle  + \sum_{t=1}^n \Big\{ trace\Big({F}_{t}^{CL} K_{I_t}({F}_{t}^{CL})^T P_{t}\Big)   \nonumber \\
& \hst  +   trace \Big(Q_t \Sigma_t+ Q_t (\Gamma_t^1)K_t + R_t K_{Z_t}\Big)       \Big\} .
\end{align}
(ii) The optimal   $\{(\Gamma_{t}^{1,*}, K_{Z_t}^*)| t\in {\mathbb Z}_+^n\}$  is determined from  (\ref{ftfic_is}), with ${\cal P}_{n}^{\widehat{\hat{X}}}(\kappa)$ replaced by $\frac{1}{n}J_{n}^{SC}(U^*, \Gamma^1, K_Z) \leq \kappa $, i.e., $U(\cdot)=U^{*}(\cdot)$.
 \end{theorem}
 \vspace*{-0.3cm}
\begin{proof} See  Section~\ref{app:LQG-DM}. 
\end{proof} 

\vspace*{-0.2cm}

\begin{remark} 
\label{rem:rel}
The characterizations of  $C_{FB,n}(\kappa)$ in  \cite{charalambous2020new,CDC-CK-SL:ITW2021} correspond to $B_t=0, Q_t=0, R_t=I, \forall t$. Using these values in   $U^*(\cdot)$ of   (\ref{opt_con_1_1_n_a}), (\ref{opt_con_1_1_n}), and $C_{FB, n}(\kappa)$ of Theorem~\ref{thm_SS},   the optimal input reduces to $A_t=  \Gamma_{t}^1\big(\hat{X}_{t}- \widehat{\hat{X}}_{t}\big) + Z_t, \forall t \in {\mathbb Z}_+^n$, with  corresponding  $C_{FB, n}(\kappa)$, which   coincides with the 
expression in \cite{charalambous2020new,CDC-CK-SL:ITW2021}. \\
Reference  \cite{charalambous2020new} discusses relations to the  expressions of $C_{FB, n}(\kappa)$  derived in  \cite{yang-kavcic-tatikonda2007,kim2006,kim2010,gattami2019}.
\end{remark} 

\label{sect_POSS}

\section{Asymptotic Limit}
 In this section we consider the limit $C_{FB}(\kappa) =\lim_{n \longleftrightarrow \infty}\frac{1}{n} C_{FB,n}(\kappa)$, based on Assumptions~\ref{ass:ati}.
 
 \begin{assumptions} Asymptotically time-invariant  \\
 \label{ass:ati}
 (1) The  LQG-POSS  (\ref{LQG-1})-(\ref{LQG-7}) is asymptotically time-invariant (ATI), $\lim_{n \longrightarrow \infty} (F_n, B_n, C_n, D_n, N_n, K_{W_n}, R_n, Q_n)=(F, B, C, D, N, K_{W}, R, Q) $. \\
(2) The  
  controller-encoder strategies are ATI,  $\lim_{n \longrightarrow \infty} (\Gamma_n^1, K_{Z_n})= (\Gamma^1,  K_{Z}), K_{Z} \succeq 0$.\\
  (3) The solutions to the matrix DREs $(\Sigma_n, K_n,  P_n)$ are such that $\lim_{n \longrightarrow \infty}(\Sigma_n, K_n,  P_n)=(\Sigma, K,  P)$, where $\Sigma\succeq 0, K \succeq 0, P \succeq 0$ are unique stabilizing solutions of corresponding matrix algebraic Riccati equations (AREs),
  \begin{align}
\Sigma=& F \Sigma F^T   -\Big(F  \Sigma C^T+GK_{W}N^T  \Big)  \Big(N K_{W} N^T+C  \Sigma C^T\Big)^{-1} \nonumber \\
&\Big( F  \Sigma C^T+ G K_{W}N^T  \Big)^T + G K_{W} G^T,\hso \Sigma \succeq 0,  \label{are_1}\\
&\mbox{similarly for $(K, P)$.} 
\end{align}
 \end{assumptions}
 
 \begin{theorem} Asymptotic char.  of  ${C}_{FB}(\kappa)$ for LQG-POSSs   \\ 
\label{thm_SS_as}
Consider  the LQG-POSS  (\ref{LQG-1})-(\ref{LQG-7}), and suppose Assumptions~\ref{ass:ati} hold. 
The CC capacity $C_{FB}(\kappa)$ is given by  
\begin{align}
&C_{FB}(\kappa) =\lim_{n \longleftrightarrow \infty}\frac{1}{n} C_{FB,n}(\kappa)= \sup_{{\cal P}_{\infty}^{\widehat{\hat{X}}}(\kappa)} \frac{1}{2}  \log \frac{ \det( K_{I})   }{\det(K_{\hat{I}})} \\
&= \sup_{{\cal P}_{\infty}^{\widehat{\hat{X}}}(\kappa)}\frac{1}{2} \log \Big\{ \frac{ \det \Big( C( \Gamma^1)K (C (\Gamma^1)^T + K_{\hat{I}} + DK_{Z}D^T \Big)    }{\det\big(K_{\hat{I}}\big)}\Big\},   \nonumber \\
&{\cal P}_{\infty}^{\widehat{\hat{X}}}\tri \Big\{ (\Gamma^1, K_Z)\big| \: K_{Z} \succeq 0, \hso   trace\Big({F}^{CL} K_{I}({F}^{CL})^T P\Big)   \nonumber \\
& \hst \hst  +   trace \Big(Q \Sigma+ Q (\Gamma^1)K + R K_{Z}\Big)  \leq \kappa    \Big\}, \nonumber\\
&\mbox{s.t. $\Sigma\succeq 0, K \succeq 0, P \succeq 0$ satisfy the matrix AREs}\nonumber 
\end{align}
where the limit is the uniform limit over all initial conditions, i.e., is independent of $ \Sigma_1, P_1, K_1$.  \\ 
Moreover, sufficient conditions for Assumptions~\ref{ass:ati}.(3)  to hold, are  (i) detectability and (ii)  stabilizability of the  AREs. 
\end{theorem}
\vspace*{-.3cm}
\begin{proof} Using Theorem~\ref{gen_exa} and Theorem~\ref{thm_SS}.(b) and    Assumptions~\ref{ass:ati}, we repeat  the proof  given in \cite{CDC-SL:ITW022} to show the limit over $n \longrightarrow \infty$ and supremum over the strategies can be interchanged,   the asymptotic limit is a uniform limit over all initial conditions.  The fact that (i) detectability and (ii)  stabilizability of the matrix AREs are sufficient for Assumptions~\ref{ass:ati}.(3) to hold follows from the continuous  dependence of solutions of matrix DREs on the parameters. 
\end{proof} 
 
%

\section{Conclusion}
In this paper we applied the concepts of  information state and sufficient statistic to characterize feedback capacity $C_{FB}$ of nonlinear partially observable stochastic systems,  with input  dependent   states, which are not available to the encoder and the decoder. In addition, we applied these concepts to linear-quadratic-Gaussian  partially observable stochastic systems (LQG-POSS). For LQG-POSS,  $C_{FB}$ is expressed in terms of 2 filtering Riccati equations and 1 control Riccati equation.  For certain  special cases we recover  recent expressions of $C_{FB}$ that appeared in \cite{charalambous2020new,CDC-CK-SL:ITW2021}, for Gaussian channels with memory, without state dependent inputs. 


\section{Appendix}

\subsection{Definition of Information State and Sufficient Statistic}
\label{app:is-suf}

\begin{definition} (Information state and sufficient statistic)\\ 
\label{def:is}
(a) The a posteriori distribution $\{{\bf P}_{t}^P(dx_t|a^{t-1}, y^{t-1})| t\in {\mathbb Z}_+^n\}$  is  called an information state if  the next state ${\bf P}_{t+1}^P(dx_{t+1}|a^{t}, y^{t})$ is determined from $(y_t, a_t)$,  the current state ${\bf P}_{t}^P(dx_t|a^{t-1}, y^{t-1})$, and possibly  ${ P}_{t}(da_t|a^{t-1}, y^{t-1})$,  $\forall t\in {\mathbb Z}_+^n$.\\
(b) A  statistic $\{\delta_t|t \in {\mathbb Z}_+^n\}, \delta_t : \Omega \rar {\mathbb D}$ is called a sufficient statistic for the strategies $P(\cdot|\cdot) \in {\cal P}_{n}(\kappa)$ if $P_i(da_t|a^{t-1}, y^{t-1})$ is  induced by $A_t=\mu_i(\delta_t)$ for some measurable function $\mu_t(\cdot)$, $\forall t \in {\mathbb Z}_+^n$. \\
(c)  The strategies $P(\cdot|\cdot) \in {\cal P}_{n}(\kappa)$ are called  separated  strategies  if they are generated by  $A_t={\mu}_t^z(\xi_t(A^{t-1}, Y^{t-1}), Z_t)$, where $\xi_t(a^{t-1}, y^{t-1})\tri {\bf P}_{t}^P(dx_t|a^{t-1}, y^{t-1})$, $\mu_t^z(\cdot)$ is a measurable function, and   $Z_t: \Omega \rar {\mathbb  Z}$ is a RV (responsible for the randomization),  
 $\forall t\in {\mathbb Z}_+^n$ (i.e., $\delta_t=(\xi_t(A^{t-1}, Y^{t-1}), Z_t)$ is a sufficient statistic).  The set of separated strategies is denoted by ${\cal P}_{n}^{sep}(\kappa)$. 
\end{definition}

\subsection{LQG-POSS}
\label{app:LQG-DM}

{\bf Proof of Lemma~III.3}
\label{app:thm_SS_P}
We show (III.47),   (III.48) last. First, (III.44) is a re-statement of the first equality   in (III.33) in  Lemma~III.1.(a).(i).  The conditional distributions (III.45), (III.46) follow directly from  (III.44), and the orthogonality of $\hat{I}_t$ and $(A^{t},Y^{t-1})$, as follows.
\begin{align*}
&{\mb P}\big\{Y_t \in dy \Big| Y^{t-1}, X^t\big\} ={\bf P}_t(dy |Y^{t-1}, A^t),   \\
=&{\bf P}_t(dy |A^t,Y^{t-1},\hat{X}^{t}), \hso  \mbox{by $\hat{X}_t= {\bf E}\big\{X_t\Big| Y^{t-1}, A^{t-1}\big\}$}\\
=&{\bf P}_t(dy |A_t,\hat{X}_{t}), \hso  \mbox{by (III.44)}. 
\end{align*}
The equivalent representation of  $\hat{\gamma}_t(A_t, \hat{X}_t)$ in the   average cost (III.51) is obtain by reconditioning,    ${\bf E}\big\{ \sum_{t=1}^n {\gamma}_t(A_t, {X}_t)    \big\}=\sum_{t=1}^n {\bf E}\big\{ {\bf E}\big\{  {\gamma}_t(A_t,{X}_t)\big|A^t, Y^{t-1} \big\}   \big\}$,  using  the quadratic definition of $\gamma_i(\cdot)$ of (I.16),  Lemma~III.1, and the Markov chain $X_t \leftrightarrow (A^{t-1}, Y^{t-1}) \leftrightarrow A_t$. This shows the new channel (III.44) is linear in  $(\hat{X}_t, A_t)$ and cost is quadratic in $(\hat{X}_t, A_t)$. Since ${\bf P}_{\hat{X}_t|\hat{X}^{t-1}, A^{t-1}}={\bf P}_{\hat{X}_t|\hat{X}_{t-1}, A_{t-1}}$, and $I^P(A^n \rar Y^n)
=  {\bf E}^P\big\{ 
\log\big(\frac{{\bf P}_t(dY_t|Y^{t-1}, A^t)}
{{\bf P}_t(dY_t|Y^{t-1})}\big) \big\}={\bf E}\big\{ 
\log\big(\frac{{\bf P}_t(dY_t|A_t, \hat{X}_t)}
{{\bf P}_t(dY_t|Y^{t-1})}\big) \big\}$,  by Markov decision theory the optimization over ${\cal P}_{n}(\kappa)$ of $I(A^n\rar Y^n)$ occurs in the set ${\cal P}_{n}^{\hat{X}}(\kappa)$, i,.e.,  (III.47).  holds. By (III.47),  ${\bf P}_t(da_t|\hat{x}_t,y^{t-1})$ is induced     by the  realization (III.48). Hence, we obtain 
 (III.50) we use the distributions  (III.45)-(III.47).  

{\bf Proof of Theorem~III.2.}
\label{app_gen_exa}
(a) In $C_{FB,n}(\kappa) $ defined by  (III.62),   the terms  $I(A_t, \hat{X}_t; Y_t|Y^{t-1})$ given by  (III.63)     do not depend on the strategy $U_t(\cdot), \forall t \in {\mathbb Z}_+^n$, hence we can apply the person-by-person optimality concepts stated under (i) and (ii). (b) We express  $\widehat{\hat{\gamma}}_t(U_t,\widehat{\hat{X}}_t, \Gamma_t^1, K_{Z_t})\big)$ defined by (III.65) as follows.  
\begin{align}
\widehat{\hat{\gamma}}_t(U_t,\widehat{\hat{X}}_t, \Gamma_t^1, K_{Z_t})=& \left[ \begin{array}{c} \widehat{\hat{X}}_t \\ U_t \end{array} \right]^T \left[\begin{array}{cc} Q_t(\Gamma_t^1) & L_t(\Gamma_t^1) \\   \big(L_t(\Gamma_t^1)\big)^T & R_t\end{array} \right] \left[ \begin{array} {c} \widehat{\hat{X}}_t \\ U_t\end{array} \right] \nonumber  \\
&+ \mbox{additional terms}.
\label{PO_tr_3_a}  
\end{align}
Then we have a linear-quadratic Gaussian stochastic optimal control problem, and its solution is the one stated under (i) \cite{caines1988}.  Since we determined $U^*(\cdot)$ it remains to solve $C_{FB,n}(\kappa)$ as stated under (ii).

\bibliographystyle{IEEEtran}

\bibliography{Bibliography_capacity}

\newpage

\end{document}